\documentclass[runningheads,envcountsame,envcountsect]{llncs}

\usepackage[english]{babel}

\usepackage[utf8]{inputenc}  
\usepackage[T1]{fontenc}    

\usepackage{verbatim}
\usepackage{amsfonts,amsmath}
\usepackage{ amssymb }
\usepackage{ stmaryrd }

\usepackage{bbm}

\usepackage{float}

\usepackage{hyperref}

\usepackage[dvipsnames]{xcolor}

\usepackage{mathtools}

\DeclareMathOperator{\id}{\mathbbm{1}}

\usepackage{xparse}
\usepackage{xspace}

\usepackage{diagbox}

\usepackage{pgf,tikz,calc}
\usetikzlibrary{arrows,automata,shapes,positioning}
\tikzstyle{every state}=[minimum size=12pt,inner sep=0pt]
\usetikzlibrary{snakes}
\usetikzlibrary{calc}
\usetikzlibrary{matrix}
\usetikzlibrary{decorations.pathmorphing}

\tikzset{
  loop above right/.style={above right, out= 60, in= 30, loop},
  loop above left/.style ={above left,  out=150, in=120, loop},
  loop below right/.style={below right, out=330, in=300, loop},
  loop below left/.style ={below left,  out=240, in=210, loop}}

\NewDocumentCommand{\range}{soms}{
  \IfBooleanTF{#1}{]}{[}%
  \IfValueT{#2}{#2{\,:\,}}%
  #3%
  \IfBooleanTF{#4}{[}{]}%
}
\NewDocumentCommand{\rangeint}{soms}{
  \IfBooleanTF{#1}{\rrbracket}{\llbracket}%
  \IfValueTF{#2}{#2}{1}
  \,:\,%
  #3%
  \IfBooleanTF{#4}{\llbracket}{\rrbracket}%
}
\newcommand\OP{{\textsc{\texttt{Order Problem}}}\xspace}
\newcommand\FP{{\textsc{\texttt{Finiteness Problem}}}\xspace}

\newcommand{\eg}{{\textit{e.g. }}}
\newcommand{\ie}{{\textit{i.e. }}}
\newcommand{\resp}{{\textit{resp. }}}

\newcommand{\set}[1]{ \left\{ #1 \right\} }
\newcommand{\size}[1]{{\left| #1 \right|}}
\newcommand{\dz}{\mathfrak d}
\newcommand{\mz}{\mathfrak m}
\newcommand{\aut}[1]{{\mathcal #1}}
\newcommand{\auta}{\aut{A}}

\newcommand{\mot}[1]{{\mathbf {#1}}}

\newcommand{\alphab}{\Sigma}
\newcommand{\XX}{\Sigma}
\newcommand{\QQ}{Q}
\newcommand{\auttuple}{(Q,\alphab,\delta,\rho)}
\newcommand{\pres}[1]{\langle#1\rangle}
\newcommand{\presp}[1]{\langle#1\rangle_+}

\newcommand{\gauta}{\pres{\auta}}

\newcommand{\Z}{{\mathbb Z}}

\newcommand{\GSk}[1]{S_{#1}}

\newcommand{\sect}[2]{{#1}_{|{#2}}}
\newcommand{\act}[2]{{#1} {\cdot} {#2}}

\newcommand{\Orb}{\text{Orb}}
\newcommand{\pred}{\text{pred}}

\newcommand{\wop}[2]{{\prescript{#1}{}{#2}}}
\newcommand{\wopa}{{\wop{\ell\!\!}{\auta}}}

\newcommand\Pol{{\mathrm{Pol}}}

\newcommand{\lacroix}{\tikz[baseline=-.5ex]{\draw[->,>=latex] (0,0) -- (4ex,0); \draw[->,>=latex] (1.8ex,2ex) -- (1.8ex,-2ex);}}

\title{Generic properties in some classes of~automaton~groups}
\author{Thibault Godin}
\institute{IECL, UMR 7502 CNRS \& Universit\'e de Lorraine\\
IMAG, UMR 5149 CNRS \& Universit\'e de Montpellier\\
 \email{thibault.godin@umontpellier.fr}}
\date{\today}

\begin{document}

\maketitle

\begin{abstract}
	We prove, for various important classes of Mealy automata, that almost all generated groups have an element of infinite order. In certain cases, it also implies other results such as exponential growth.
\end{abstract}


\section{Introduction}

The class of groups generated by Mealy automata presents a considerable variety of behaviours and has been widely used since the eighties as a powerful source of interesting groups~\cite{Gri80,Gri84,Wil04,BaVi05}. It seems natural to try to produce new examples of groups to be studied by picking a random Mealy automaton and considering the group it generates, or to try  to get an interesting distribution over some class of groups starting from a distribution over some class of Mealy automata~\cite{God17}. This approach also raises a natural question: \emph{"how does a typical automaton group look like?"}. In this paper, we tackle this problem and give partial answers for several important and well-studied classes, by proving that automata belonging to the class of reversible, reset, or polynomial activity automata generate with great probability a group having at least one element of infinite order. In particular, it means that these groups are generically infinite and not Burnside.

Another motivation for this paper is that the \OP---how to decide whether an element generates an infinite group---was recently proven undecidable among automaton groups~\cite{Gil18,BaMi17}, while the~\FP--how to decide whether the whole group is infinite---is known to be undecidable for automaton semigroups but remains open for automata groups~\cite{Gil14,DeOl17}. On the other hand, some classes of automaton (semi)groups are known to have decidable \OP~\cite{BBSZ13,BGKP18}. Our results provide statistical answers for these problems.

Depending on the class, we also get stronger or additional statements, among others, the groups generated by reversible or reset Mealy automata have generically exponential growth.\smallskip

The proposed proofs vary strongly with the considered class and rely on the structural properties of the automata. In particular, the case  of general  invertible Mealy automata remains open. 

\smallskip

In order to simplify the statements, we will use the informal "let~\(\auta\) be a random automaton in~\(\mathcal{C}\)" instead of the  formal "let~\(\auta\) be a random variable uniformly distributed over the set~\(\mathcal{C}\)". All probabilistic statements should be understood accordingly.


\section{Automaton groups}

We recall that the \emph{order} of an element~\(g\) of a group~\(G\) is the least (strictly positive) integer~\(\alpha\) such that~\(g^\alpha = \id\). If such an integer does not exist, we say that~\(g\) has infinite order. Equivalently, the order of~\(g\) is the cardinal of the subgroup it generates, hence having an element of infinite order implies the infiniteness of the whole group.

If~\(X\) is a finite set then~\(X^k\) denotes the set of words of length~\(k\), and~\(X^*\) (\resp~\(X^+\)) the set of words of arbitrary (\resp positive) length. We take as a convention that elements of~\(X^\ell, \:\ell>1\) are represented with a bold font.

\subsection{Mealy automata and automaton (semi)groups}

A \emph{Mealy automaton} is  a 4-tuple~\(\auta = \auttuple\) where~\(\QQ\) and~\(\XX\) are finite sets, called the \emph{stateset} and the \emph{alphabet},~\(\delta = {\set{\delta_x : \QQ \to \QQ}}_{x \in \XX} \) is a set of functions called \emph{transition} functions, and~\(\rho = {\set{\rho_q : \XX \to \XX}}_{q \in \QQ} \) is a set of functions called \emph{production} functions. Examples of such  automata are presented on Figure~\ref{fig:Mealy}, and we refer the reader to~\cite{Nek05} for a more complete introduction.\\
\begin{figure}[h!]%
\begin{center}
{
\scalebox{1.125}{
\footnotesize{
	\begin{tikzpicture}[->,>=latex,node distance=12mm]

		\node[state] (c)  {$c$};
				\node[state] (a) [left of=c,below of=c] {$a$};
				\node[state] (b) [left of=c, above of=c] {$b$};
				\node[state] (d) [right of=c, above of=c] {$d$};
				\node[state] (e) [right of=c,below of=c] {$\id$};
			\path
		
				(a)	edge	node[below=.1cm]{\(\begin{array}{c} {0|1} \\ {1|0} \end{array}\)}	(e)
				(b) edge 	node[left=.1cm]{\(0|0\)} (a)
							(b) edge 	node[below=.1cm, pos=0.4]{\(1|1\)} (c)
				(c) edge 	node[above=.1cm, pos=0.6]{\(0|0\)} (a)
				(c) edge 	node[above=.1cm, pos=0.25]{\(1|1\)} (d)
				(d) edge 	node[right=.1cm]{\(0|0\)} (e)
				(d) edge 	node[above=.1cm]{\(1|1\)} (b)
				(e) edge[loop below right] node[below right =-5mm and -1mm] {\(\begin{array}{c} 0|0\\ 1|1 \end{array}\)}	(e)
		
		;
		\scalebox{1}{
		\footnotesize{
\begin{scope}[xshift=3cm, ->,>=latex,node distance=16mm]

	\node[draw,circle,minimum width=19pt,inner sep=0pt] (t) {$p$};
	\node[draw,circle,minimum width=19pt,inner sep=0pt,right of=t] (one) {$\id$};
	\path	

		(t)	edge[loop above] 		node[above]{\(1|0\)}	(t)
		(t)	edge				node[below]{\(0|1\)}	(one)
		(one)	edge[loop right]		node[right,align=center]{\(0|0\) \\\(1|1\)}	(one);

\end{scope}
	}}
	\end{tikzpicture}
	}
	}	
}
\end{center}
\caption{The automaton generating  the Grigorchuk group (left) and  the adding machine, generating~\(\Z\) (right). Both have bounded activity, and~\(p\) has infinite order.}\label{fig:Mealy}
\end{figure}
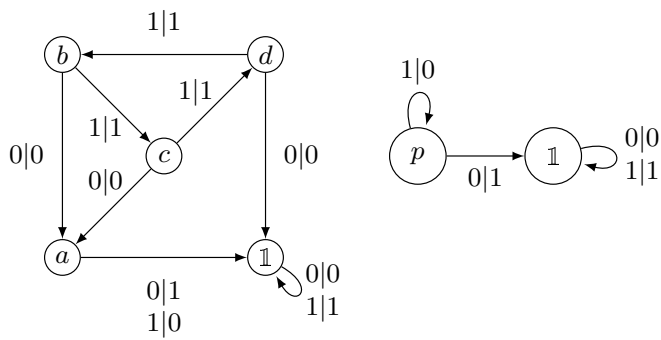
The map~$\rho_q$ extends to a length-preserving map on words~$\rho_q\colon\XX^* \rightarrow \XX^*$ by the recursive definition:
\begin{equation*}
\forall i \in \XX, \ \forall \mot{s} \in \XX^*, \qquad
\rho_q(i\mot{s}) = \rho_q(i)\rho_{\delta_i(q)}(\mot{s}) \:.
\end{equation*}
We can also extend the set of maps $\rho$ to words of states $\mot{u} \in \QQ^*$ by composing the production functions associated with the letters of $\mot{u}$:
\begin{equation*}
 \forall q \in \QQ, \ \forall \mot{u} \in \QQ^*, \qquad
\rho_{q\mot{u}} = \rho_{\mot{u}}\circ \rho_{q}\:.
\end{equation*}

Likewise, we extend the functions~\(\delta\) to words of state and words via
\[\forall i \in \XX, \ \forall \mot{s} \in \XX^*, \ \forall q \in \QQ, \ \forall \mot{u} \in \QQ^*, \quad
\delta_i(q\mot{u}) = \delta_i(q)\delta_{\rho_q(i)}(\mot{u})  \text{ and } \delta_{i\mot{s}} = \delta_i \circ \delta_{\mot{s}}\:.\]
\hspace{-.1cm}
\begin{minipage}[t]{0.46\linewidth}   
For each automaton transition $q~\xrightarrow{x|\rho_q(x)}~\delta_x(q)$, we associate the \emph{cross-transition} depicted in the following way:
\end{minipage}
\hfill
\begin{minipage}[t]{0.55\linewidth}   
\vspace{-.4cm}
\[\begin{array}{ccc}
		& x            	&    \\
q	& \lacroix    	& \delta_x(q)\\
		& \rho_q(x) 			&            		
\end{array}\]
\end{minipage}
\smallskip

The production functions $\rho_q : \XX^* \to \XX^* $ of an automaton $\auta$ generate a semigroup $\presp{\auta}= \{\rho_{\mot{u}} : \XX^* \to \XX^* | \mot{u} \in \QQ^+\}$.

A Mealy automaton is \emph{invertible} when the functions~\(\rho\) are permutations of~\(\XX\). When a Mealy automaton is invertible one can define its \emph{inverse}~\(\auta^{-1}\) by~\[p \xrightarrow{x\mid y} q \ \in \auta \Leftrightarrow p^{-1} \xrightarrow{y\mid x} q^{-1} \in \auta^{-1}\:.\]

Whenever a Mealy automaton is invertible we can consider the~\emph{group}~\(\gauta\) it generates:
\[\pres{\auta} =  \pres{\rho_q \ | \ q \in \QQ} = \set{\rho^{\pm 1}_{\mot{u}}\mid \mot{u}\in \QQ^*}\:.\] 

A group (\resp a semigroup) is an \emph{automaton group} (\resp \emph{semigroup}) if it can be generated by some Mealy automaton.

Given a Mealy automaton~\( \auta = \auttuple\), its \emph{dual} is the Mealy automaton~\(\dz \auta  = \left(\XX,\QQ,\rho,\delta \right)\) where the roles of the stateset and of the alphabet are exchanged. Its~\(\ell\)-th power is the automaton~\(\left(\QQ^\ell,\XX,\delta,\rho \right)\)where the production and transition functions have been naturally extended. We define also the automaton~\(\wopa = \left(\QQ,\XX^\ell,\delta,\rho \right) = \dz(\dz \auta)^\ell\) acting on sequences of~\(\ell\) letters and remark that this operation does not change the generated semigroup, \ie ~\(\presp{\wopa} = \presp{\auta}\).\\

\smallskip

From an algebraic point of view, it is convenient to describe the elements of an automaton group via the so-called wreath recursions. For any~\(g\) in an automaton group~\(\gauta\) on alphabet~\(\XX\) and any word~\(\mot{s} \in \XX^*\), let \(\act{g}{\mot{s}} 
\) denotes the image of~\(\mot{s}\) by~\(g\), and \(\sect{g}{\mot{s}} 
\) the unique~\(h \in \gauta\) satisfying \(\act{g}{(\mot{st})}=(\act{g}{\mot{s}})\act{h}{\mot{t}}\)
for all~\(\mot{t}\in\XX^*\).
The \emph{wreath recursion} of~\(g\) is:
\[g=(\sect{g}{x_1},\ldots,\sect{g}{x_{|\XX|}})\sigma_g,\]
where~\(\sigma_g \in \GSk{\size{\Sigma}}\) denotes the permutation on~\(\XX\) induces by~\(g\).\smallskip

\begin{figure}[h]%
\centering
\begin{tikzpicture}[->,>=latex,node distance=30mm]
		\node[state] (b) {$b$};
		\node[state] (a) [right of=b] {$a$};
		\node[state] (f) [right of=a] {$f$};
		\node[state,node distance=20mm] (c) [below of=b] {$c$};
		\node[state] (d) [right of=c] {$d$};
		\node[state] (e) [right of=d] {$e$};
	\path
		(a)	edge	node[below]{\(1|3\)}	(b)
		(a)	edge[bend right,out=340,in=200]	node[left]{\(2|2\)}	(d)
		(a)	edge	node[below]{\(3|1\)}	(f)
		(b)	edge	node[below, very  near start]{\(3|1\)}	(d)
		(b)	edge[bend left,out=70,in=110]	node[below = -0.1cm]{\(\begin{array}{c}1|3\\ 2|2\end{array}\)}	(f)
		(c)	edge	node[above, very near start]{\(1|3\)}	(a)
		(c)	edge	node[left]{\(3|1\)}	(b)
		(c)	edge[bend right, out=335,in=205]	node{\(\begin{array}{c}2|2\\~\end{array}\)}	(e)
		(d)	edge[bend right, out=340,in=200]	node[right]{\(2|2\)}	(a)
		(d)	edge	node[above]{\(3|1\)}	(c)
		(d)	edge	node[above ]{\(1|3\)}	(e)
		(e)	edge	node[above, very near start]{\(3|1\)}	(a)
		(e)	edge[bend left,out=70,in=110]	node[above= -0.1cm]{\(\begin{array}{c}1|3\\ 2|2\end{array}\)}	(c)
		(f)	edge[bend right,out=335,in=205]	node{\(\begin{array}{c}\\2|1\end{array}\)}	(b)
		(f)	edge	node[below, very  near start]{\(1|3\)}	(d)
		(f)	edge	node[right]{\(3|2\)}	(e);
\end{tikzpicture}
\caption{A 3-letter 6-state invertible  reversible non-bireversible Mealy automaton. } %

\label{fig-jir36}
\end{figure}
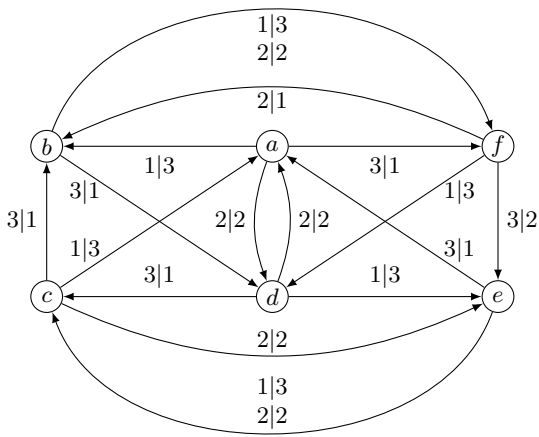

\medskip
In what follows,~\(\auta\) will denote, if not explicited, an invertible Mealy automaton~\(\auttuple\).

\subsection{Classes of Mealy automata}

We now describe several important classes of (invertible) Mealy automata.\smallskip

An automaton~\(\auta=\auttuple\) is called \emph{reversible} when the functions~\(\delta\)  are permutations of~\(\QQ\). If an automaton is invertible then its dual is reversible. A Mealy automaton is \emph{bireversible} if it is invertible and reversible, and so is its inverse. \smallskip

Another, somewhat opposite, restriction on the transition function leads to the class of \emph{reset} automata, studied \eg in~\cite{SiSt05,DeOl17}. An automaton~\(\auta\) is called~\emph{reset} is there exists a function~\(\phi: \XX \to \QQ\) such that~\(\forall x, \forall q, \delta_x(q) = \phi(x)\). In other words,  all the arrows labelled by an input letter~\(x\) lead to the same state~\(\phi(x)\). Up to renaming the states and pruning the automaton of its vertices without ingoing edges (which does not change the finiteness of the generated group nor the existence of element of infinite order), we may assume that all studied reset automata are~\emph{unfolded}, \ie that~\(\QQ=\XX\) and~\( \phi=\id \).
\smallskip

Another class of Mealy automata linked to the cycle structure is defined in~\cite{Sid00}, via the~\emph{activity}. Assume that there is a unique state inducing the identity in the group, denoted~\(\id\). The activity of an automorphism~\(t \in \gauta\) is defined as the function
\[\alpha_t : \ell \mapsto \size{\set{ x \in \XX^\ell, \: \sect{t}{x}\neq \id}} \:.\]
It is known that the activity~\(\alpha_t\) is polynomial if and only if there is not two nontrivial simple cycles accessible one from another in the automaton~\(\auta\), and that in this case the degree of the polynomial is the maximal number of nontrivial cycles that can be reached along of a simple cycle minus one. For a fixed alphabet~\(\XX\), we denote~\(\Pol\) (\resp~\(\Pol(d)\)) the set of all Mealy automata with polynomial activity (\resp with activity bounded by a polynomial of degree~\(d\)), and in particular we call \emph{bounded} (\resp \emph{finitary}) the set~\(\Pol(0)\) (\resp~\(\Pol(-1) = \set{ t,  \: \alpha_t(\ell) \to 0}\)). Notice that every finitary automaton generates a finite group \cite{Rus10}.\medskip

In~\cite{BBSZ13} is defined a tool to understand the orbits of elements of~\(\XX^*\) under the action of an automorphism described by a Mealy automaton. Given~\(\auta=\auttuple\),~\(t \in \gauta\) and~\(\mot{x} \in \XX^*\) define~\(\Orb_t(\mot{x})  = \min_{\alpha>0} \set{\act{t^\alpha}{\mot{x}}=\mot{x}}\) the size of the orbit of~\(\mot{x}\) under the action of~\(t\). The \emph{Orbit Signalizer} is the graph~\(\Gamma_t\) whose vertices  are the~\(\sect{t^{\Orb_t(\mot{x}) }}{\mot{x}}, \: \mot{x} \in \XX^*\)  and edges from \(\sect{t^{\Orb_t(\mot{x})}}{\mot{x}}\) to \(\sect{t^{\Orb_t(\mot{x}y)}}{\mot{x}y}, \mot{x} \in \XX^*, y \in \XX\) with label~\(\Orb_{\sect{t^{\Orb_t(\mot{x})}}{\mot{x}}}(y)\).\\
The Orbit signalizer is used in~\cite{BBSZ13,BGKP18} to solve the \OP. Indeed the order of~\(t\) is the lowest common multiple of all labels along paths starting from vertex~\(t\) in~\(\Gamma_t\). In particular if the orbit signalizer is finite then the \OP is decidable for~\(t\), as it reduces to checking if cycles have labels all~\(1\).


\section{Reversible Mealy automata}
We show that groups generated by (invertible) reversible Mealy automata have an element of infinite order with high probability. In fact we are going to prove a stronger result by showing that almost all invertible reversible automata are not bireversible, then using known results from~\cite{GKP15}, we obtain that the generated semigroups are almost surely  torsion-free.\smallskip

Since a Mealy automaton is completely defined by its transition and production functions, an invertible reversible Mealy automaton can be understood as~\(\size{\QQ}\) permutations in~\(\GSk{\size{\XX}}\) and \(\size{\XX}\) permutations in~\(\GSk{\size{\QQ}}\), thus the uniform distribution on the set of invertible reversible Mealy automata with stateset~\(\QQ\) and alphabet~\(\XX\) is the uniform distribution on~\(\GSk{\size{\QQ}}^{\size{\XX}}\times\GSk{\size{\XX}}^{\size{\QQ}}\).\smallskip

An invertible reversible automaton is bireversible if and only if each output letter induces a permutation of the stateset. In particular, for bireversible automata, we have that : \[\forall (p,i), (q,j) \in \QQ\times \XX, (p,i)\neq (q,j), \: \delta_i(p) = \delta_j(q) \Rightarrow \rho_p(i)\neq\rho_q(j)\:.\]

Define, for~\(r\in\QQ\), the set~\(\mathcal{O}_r= \set{j \in \XX, \exists (p,i) \in \QQ\times\XX,  p \xrightarrow{i|j} r \in \auta}\) of output letters that lead to~\(r\). An invertible reversible automaton is bireversible if and only if, for all states~\(r\), the set~\(\mathcal{O}_r\) is the whole alphabet.

\begin{example}
	Consider the  automaton Fig.~\ref{fig-jir36}. We have~\(\mathcal{O}_a = \set{1,2,3} =\mathcal{O}_c\) but~\(\mathcal{O}_b = \set{1,3}\), hence the automaton is not bireversible.
\end{example}

\begin{proposition}\label{prop:birev}
The probability that a random  invertible reversible automaton  with \(k\) letters and \(n\) states  is bireversible is less than \[ {1 \over n^{k-1}} + { 1\over k} \]
\end{proposition}
\begin{proof}
Let~\(\auta = \auttuple\). For~\(r \in \QQ\), we denote~\(\pred_r=\size{\set{p \in \QQ \mid \:\exists i \in \XX, \delta_i(p)=r}}\) the size of the set of predecessor of~\(r\). 
Let us fix a state \(r\). We have:
\begin{align*}
\Pr(\auta \in \text{BIR}) &= \Pr(\forall q \in \QQ, \mathcal{O}_q=\XX) \\
 &\leq \Pr(\mathcal{O}_r = \XX)\:.\\
\intertext{From the law of total probability we get:}
\Pr(\auta \in \text{BIR})&\leq \Pr(\mathcal{O}_r = \XX \mid \pred_r = 1 )\Pr( \pred_r = 1 )  \quad \\
& \qquad \quad + \quad \Pr(\mathcal{O}_r = \XX \mid  \pred_r \geq 2 )\Pr(  \pred_r \geq 2) \\
&\leq \Pr( \pred_r = 1 ) + \Pr(\mathcal{O}_r = \XX \mid \pred_r \geq 2)
\intertext{ The probability that  \(r\) has exactly one predecessor can be seen as fixing  \(\delta_i^{-1}(r)\) for some reference letter \(i \in \XX\) and requiring that the \(k-1\) other \(\delta_j^{-1}(r), \: j \in \XX \setminus \lbrace i\rbrace\) are equal to   \(\delta_i^{-1}(r)\), hence: }
\Pr(\auta \in \text{BIR})&\leq {1 \over n^{k-1}} + \Pr(\mathcal{O}_r = \XX \mid  \pred_r  \geq 2)
\intertext{For the second term, let us consider a predecessor \(p\) of \(r\) and let \( \lambda\) be the number of input letters leading from \(p\) to \(r\). We have \(1\leq\lambda \leq {k-1}\). To enforce bireversibility, we have to  avoid that~\(p\) outputs a letter that is already leading to~\(r\) (\(\rho_p(i) \neq \rho_q(j)\) for~\(\delta_p(i)=\delta_q(j)=r\)). Assume that  the set \(\mathcal{O}_r^{(p)}= \set{j \in \XX  \mid \: \exists q \neq p \in \QQ,  q \xrightarrow{i|j} r \in \auta}\) of output letters leading to~\(r\) from a state~\(q\) different from~\(p\) is of maximal size~\(k-\lambda\). Since~\(\rho_p\) is random and independent from the others~\(\rho_q\) we can bound this probability from above: having the letters leading from~\(p\) to~\(r\) produce the~\(\lambda\) out of~\(k\) required letters is~\({k \choose \lambda}^{-1}\). Hence:}
\Pr(\auta \in \text{BIR})&\leq {1 \over n^{k-1}} +  {k \choose \lambda}^{-1}\\
&\leq {1 \over n^{k-1}} +  {1 \over k} \:.
\end{align*}
\vspace{-.3cm}
\qed
\end{proof}

The obtained bound goes to~\(0\) as~\(k\) goes to infinity, and the same reasoning applied to the dual automaton allows us to conclude that the proportion of bireversible automata among the invertible reversible automata goes to~\(0\) as the number of letters or states grows.

It is proven in~\cite{GKP15}
that an invertible reversible Mealy automaton without bireversible component generates a torsion-free semigroup. Whence our theorem:
\begin{theorem}\label{thm:orderbirev} 
The probability that an invertible reversible Mealy automaton taken uniformly at random generates a  torsion-free semigroup goes to~\(1\) as the size of  the alphabet grows. Moreover, the probability for the group to have an element of infinite order also goes to~\(1\)  as  the stateset or  the alphabet grows.
\end{theorem}
\begin{proof}
	It is known that, with great probability, two random permutations on a large set generate a transitive group~\cite{Dix69}, in terms of a graph, it means that a typical reversible Mealy automaton on a large alphabet is (strongly) connected, and is not bireversible by Proposition~\ref{prop:birev} whence the first part of the result. The second part comes from the remark on the dual. \qed
\end{proof}

From~\cite{FrMi18}, where it is shown that having an element of infinite order implies exponential growth among groups generated by invertible reversible automata, we obtain:
\begin{theorem}\label{thm:reversiblegrowth}
The probability that an invertible reversible Mealy automaton taken uniformly at random generates a group with exponential growth goes to~\(1\) as the size of the stateset or of the alphabet grows.
\end{theorem}

\begin{remark}\label{rem:growthinfinite}
	Notice that Theorem~\ref{thm:orderbirev} is not \emph{a priori} a consequence of Theorem~\ref{thm:reversiblegrowth}: there exists infinite Burnside group with exponential growth (\cite{Adi79}). However, no example of such a group is known within the class of automaton groups.
\end{remark}

It is worthwhile noting that it is unknown whether the \OP is decidable within the class of (semi)groups generated by reversible Mealy automata.


\section{Reset Mealy automata}
The class of reset automata is of particular interest since it is linked to \emph{one-way cellular automata}, and  were used by Gillibert to prove the undecidability of the \OP for automata semigroups~\cite{Gil14}. For groups generated by (invertible) reset Mealy automata the \OP remains open~\cite{DeOl17}. \smallskip

As the transition function is trivial in a (unfolded) reset Mealy automaton, the uniform distribution on the set of unfolded invertible reset Mealy automata with stateset~\(\QQ\) and alphabet~\(\XX\) is the uniform distribution on~\(\GSk{\size{\XX}}^{\size{\QQ}}\).\smallskip

We are going to use a result from~\cite{Pic17}: let~\(\auta=\left( \QQ,\delta,\rho \right)\) be a (unfolded) reset automaton and define the transformation~\(\pi_\auta : q \mapsto \rho_q^{-1}(q)\) for all~\(q \in \QQ\).
\begin{theorem}[{\cite[Theorem 1.20]{Pic17}}]\label{thm:Pic}
	Let~\(\auta\) be a reset automaton. If~\(\pi_\auta\) is not a permutation then the group generated by~\(\auta\) has an element of infinite order.
\end{theorem}
We give the proof for the sake of completeness.
\begin{proof}
	If~\(\pi_\auta\) is not a permutation, then there exists~\(x_0\) which does not belong to any cycle of \(\pi_\auta\) and such that~\(\pi_\auta(x_0)=x_1\) belongs to a cycle~\(x_1 \xrightarrow{} \cdots \xrightarrow{} x_\ell   \xrightarrow{} x_1 \) of~\(\pi_\auta\). Computing the orbit of~\(x_0(x_1\cdots x_\ell)^\alpha\) under the action of any given state~\(q \in \QQ\)  gives:
	\[\begin{array}{cccccccccccccc}
	&x_0	&			& x_1 &	&x_2	& &  \cdots	& &  x_\ell & & x_1& \cdots\\	
	
q	&\lacroix& x_0&\lacroix& x_1	&\lacroix  & x_2 & \cdots & x_{\ell-1} &\lacroix & x_\ell&\lacroix & x_1\\
	
	&y_1	&		 & x_0&	&x_1	&		 & \cdots& & x_{\ell -1}& & x_\ell&\\	
	
q&\lacroix&  y_1&\lacroix& x_0&\lacroix& &  \cdots&   x_{\ell-2} &\lacroix & x_{\ell-1}&\lacroix & x_\ell\\
	
	&y'_1&			 & y'_2	&  & x_0 &		&\cdots & 	 & x_{\ell-2} & & x_{\ell-1}\\
	\vdots&\vdots&			 & \vdots	&  & \vdots &		&& 	 &\vdots & & \vdots\\
\end{array}
\]
Since~\(x_1=\rho_{x_0}^{-1}(x_0) = \rho_{x_\ell}^{-1}(x_\ell)\). So~\(\act{q^{\alpha i}}{x_0(x_1\cdots x_\ell)^\alpha} = \mot{u}{x_0(x_1\cdots x_\ell)^{(\alpha-i)}}\), for some~\(\mot{u} \in \QQ^{i\alpha}\), hence~\(q\) has infinite order.
\qed
\end{proof}

\begin{theorem}\label{thm:reset}
The probability that a random (unfolded) reset automaton on~\(k\)	letters has an element of infinite order is at least~\(1-e\sqrt{k}e^{-k}\).
\end{theorem}
\begin{proof}
	Since the~\(\rho_q ,\; q\in \XX\) are random permutations, the function~\(q \mapsto \rho_q^{-1}(q)\) can be considered has a random mapping from~\(\XX\) to~\(\XX\), and the number of permutations among mappings is~\({k! }\over{k^k}\). We conclude using Stirling's approximation and known bounds.
\qed
\end{proof}
Using~\cite{Olu17}, where Olukoya proves that groups generated by reset automata are either finite or have exponential growth, we get (see also Remark~ \ref{rem:growthinfinite}):
\begin{theorem}\label{thm:resetgrowth}
The probability that a random (unfolded) reset automaton on~\(k\)	letters has exponential growth is at least~\(1-e\sqrt{k}e^{-k}\).
\end{theorem}

From Delacourt and Ollinger~\cite[Proposition 1]{DeOl17}, our result also means that permutive one way cellular automata are generically aperiodic.

\begin{remark}
An unfolded reset automaton is minimal (\cite{AKLMP12}) if and only if each states induces a different permutation on letters. By the birthday problem, we can extend our result a bit: a random minimal unfolded reset automaton generically generates a group with exponential growth and elements of infinite order.
\end{remark}

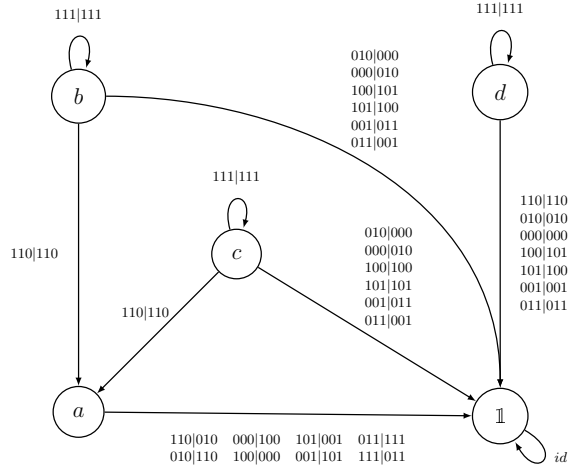
\begin{figure}[h!]

\begin{center}

{
\scalebox{.6}{
{
	\begin{tikzpicture}[->,>=latex,node distance=35mm,thick, inner sep =8pt]

		\node[state, inner sep =8pt] (c)  {\Large $c$};
				\node[state, inner sep =8pt] (a) [left of=c,below of=c] {\Large $a$};
				\node[state, inner sep =8pt] (b) [left of=c, above of=c] {\Large $b$};
				\node[state, inner sep =8pt] (d) [above right= 27.5mm and 50mm of c ] {\Large $d$};
				\node[state, inner sep =8pt] (e) [below right= 27.5mm and 50mm of c ] {\Large $\id$};
			\path
				(d) edge[loop above] node[above]{\(\begin{array}{c}  111|111 \end{array}\)}	(d)
				(c) edge[loop above] node[above]{\(\begin{array}{c}  111|111 \end{array}\)}	(c)				
				(b) edge[loop above] node[above]{\(\begin{array}{c}  111|111 \end{array}\)}	(b)				
				(a)	edge	node[below=.1cm]{\(\begin{array}{c} 110|010 \quad 000|100  \quad  101|001 \quad 011|111\\010|110 \quad 100|000  \quad  001|101 \quad 111|011 \end{array} \)}	(e)
				(b) edge 	node[left=.1cm]{\(\begin{array}{c} 110|110 \\ \end{array}\)} (a)
				(c) edge 	node[above=.2cm, pos=0.6]{\(\begin{array}{c} 110|110 \end{array}\)} (a)
				(b) edge[out=00,in=90] 	node[above=.1cm]{\(\begin{array}{c} 010|000 \\000|010 \\ 100|101 \\ 101|100\\ 001|011 \\011|001 \end{array}\)} (e)
				(c) edge 	node[above=.1cm, pos=0.6]{\(\begin{array}{c}  010|000 \\000|010 \\ 100|100 \\ 101|101\\ 001|011  \\011|001\end{array}\)} (e)
				(d) edge 	node[right=.1cm]{\(\begin{array}{c}110|110 \\ 010|010 \\000|000 \\ 100|101 \\ 101|100\\ 001|001 \\011|011\end{array}\)} (e)
				(e)  edge[loop below right] node[below right =-5mm and -1mm]{\(\begin{array}{c} id \end{array}\)}	(e)
		
		;
	\end{tikzpicture}
	}
	}	
}
\end{center}

\caption{The normal form~\(\wop{3}{\aut{G}}\) of the (invertible) automaton~\(\aut{G}\) of bounded activity generating the Grigorchuk group. A simple path contains at most one nontrivial loop.}\label{fig:Gri3}
\end{figure}

\section{Mealy automata with polynomial activity}

The class of Mealy automata with polynomial activity is interesting as the \OP is decidable for (semi)groups generated by automata with bounded activity but remains open for the higher levels of the hierarchy~\cite{BBSZ13}.\\
Recall that~\(\id\) denotes the identity state in the automaton, which is supposed to be unique.
We are going to define a \emph{normal form}: let~\(\auta = \auttuple\) be an automaton with polynomial activity and let~\(\ell\) be the lowest common multiple of the sizes of the (simple) cycles. Since an automaton with polynomial growth has no entangled cycles, we have that~\(\wopa\) has all cycles of length one. Now put~\(d\) the maximal length of an (oriented) path between a state and a self-loop in~\(\wopa\). Then the normal form of the automaton~\(\auta\) is the automaton~\(\wop{d\ell\!\!}{\auta}\) and it looks as follows: it is a directed acyclic graph whose leaf induces the identity~\(\id\), and where each state either has a self-loop or leads to a state with a self-loop. For instance the normal form~\(\wop{3}{\aut{G}}\) of the (invertible) Grigorchuk automaton~\(\aut{G}\)  (which has bounded activity) Fig.~\ref{fig:Mealy} (left) is depicted Fig.~\ref{fig:Gri3}\smallskip

To the best of our knowledge, and even among automata under normal form, there is no easy description of the uniform distribution on the set of (invertible) Mealy automata with polynomial activity, even if one fixes the degree of the activity. To bypass this difficulty, we show that automata with finitary activity are rare even among  automata with bounded activity, and use the fact that once the transition functions are fixed, the choice of production functions does not change the activity.\smallskip

The next proposition is a simple yet useful observation:
\begin{figure}[h!]

\begin{center}
\tikzset{node style ge/.style={circle, inner sep =0mm}}

\begin{tikzpicture}[baseline=(A.center),scale=.8]

  \tikzset{BarreStyle/.style =   {opacity=.4,line width=5 mm,line cap=round,color=#1}}
    \tikzset{SignePlus/.style =   {above left,,opacity=1,circle,fill=#1!50}}
    \tikzset{SigneMoins/.style =   {below left,,opacity=1,circle,fill=#1!50}}
\matrix (A) [matrix of math nodes, nodes = {node style ge},row sep=0mm,column sep=0 mm] 
{
{} &1	&		{}	& 1 & {}	&1	&		{}	&  &{} &  & {}\\		
t	&{\lacroix}&t &{\lacroix}& t	&{\lacroix}& t&{}& {}\\	
{}	&2	&	{}	 &{}&{}	&{}	&	{}	 & & {}& & {}\\		
t	&{\lacroix}&  r_i&{2}&  	&{ }& {}&{ }& {}\\	
\vdots	& {}&{} &{\lacroix}	&  {\id}& {2}&		{}	 & {}& & {} & {}\\
t	&{\lacroix}& r_j&{} & 	{}&{\lacroix}& {\id}&{}& {}\\	
{}&	1 &	{}	& k&{} &	{} &		{}& {} \\
t^{\alpha_1}	&{\lacroix}& ts_1&{\lacroix}& r_k	&{}& {}&{}& {}& {} & {}\\	
{} &{}	&		{}	& . & {}	&{}	&		{}	& {} &{} & \\		
\vdots	&{}& {}&{}& {}	&{}& {}&{}& {}& {} & {}\\	
{} &1	&		{}	& \ell & {}	&{}	&		{}	& {} &{} & \\		
t^{\alpha_1}		&{\lacroix}& ts_1&{\lacroix}& r_\ell	&{}& {}&{}& {}& {} & {}\\	
{} &1	&		{}	& 1 & {}	&m	&		{}	&  &{} & \\			
t^{\alpha_2}		&{\lacroix}& (ts_1)^{\beta_1}&{\lacroix}& \:ts_2\:\:	&{\lacroix}&r_m&{}& {}&  & {}\\	
{} &1	&		{}	& 1 & {}	&.	&		{}	&  &{} & \\		
\vdots	&{}& &{}& &{}& &{}& {}& & {}\\	
{} &1	&		{}	& 1 & {}	&n	&		{}	&  &{} & \\			
t^{\alpha_2}	&{\lacroix}&  (ts_1)^{\beta_1}&{\lacroix}& ts_2	&{\lacroix}& r_n&{}& {}&  & {}\\	
{} &1	&		{}	& 1 & {}	&1	&		{}	&  &{} & \\		
};

\draw[gray,decorate,decoration={brace,amplitude=5pt,aspect=.5,angle=20}]  ($(A-6-1.south west)+(-.2,0)$) -- ($(A-2-1.north west)+(-.2,0)$) ;
\node[black, dashed , fill =white,draw, rounded corners, inner sep=1mm] at ($(barycentric cs:A-6-1=1,A-2-1=1) -(1,0)$) {\textcolor{black}{$\alpha_1$}};

\draw[gray,decorate,decoration={brace,amplitude=5pt,aspect=.5,angle=20}]  ($(A-12-1.south)-(1.5,0)$) -- ($(A-2-1.north)-(1.5,0)$) ;
\node[black, dashed , fill =white,draw, rounded corners, inner sep=1mm] at ($(barycentric cs:A-12-1=1,A-2-1=1) -(2.33,0)$) {\(\begin{array}{c}\alpha_2 \\= \\ \alpha_1\beta_1\end{array}\)};

\draw[gray,decorate,decoration={brace,amplitude=5pt,aspect=.5,angle=20}]   ($(A-2-9.north)+(3.5,0)$) -- ($(A-18-9.south)+(3.5,0)$)  ;
\node[black, dashed , fill =white,draw, rounded corners, inner sep=1mm] at ($(barycentric cs:A-18-9=1,A-2-9=1) +(4.25,0)$) {\(\begin{array}{c}\alpha_3\end{array}\)};

\draw [BarreStyle=gray,dotted] (A-3-3.north)  to (A-6-3.south) ;
\node[below right= 1.2cm and 0.1cm of A-3-3, black, dashed , fill =white,draw, rounded corners, inner sep=1mm] {\textcolor{black}{$s_1$}};
\draw [BarreStyle=gray] (A-3-5.north)  to (A-12-5.south) ;
\node[below right= .65cm and -0.25cm of A-5-5, black, dashed , fill =white,draw, rounded corners, inner sep=1mm ] {\textcolor{black}{$s_2$}};
\draw [BarreStyle=gray] (A-3-7.north)  to (A-18-7.south) ;
\node[below right= .55cm and -0.cm of A-10-7, black, dashed , fill =white,draw, rounded corners, inner sep=1mm ] {\textcolor{black}{$s_3$}};
\node[below right= .42cm and 0.6cm of A-10-7, inner sep=1mm ] {\textcolor{black}{$=\sect{t^{{\Orb_t(111)}-1}}{222}$}};

\end{tikzpicture}

\end{center}

\caption{Computation of the orbit of~\(1^j\) under the action of~\(t\) in a bounded automaton.}\label{fig:orbit}
\vspace{-0.6cm}
\end{figure}

\begin{proposition}\label{prop:boundorder}
	Let~\(\auta=\auttuple\) be a Mealy automaton with bounded activity under normal form. If there is some~\(t \in \QQ\) with~\(\delta_{i}{(t)} = t\) and~\(\rho_{t}{(i)} \neq i\), then \(\rho_{t}\) has infinite order.
\end{proposition}
\begin{proof}
	Up to renaming, we can assume that~\(i=1\) and \(\rho_{t}{(1)} =2\). We  use the orbit signalizer of~\(t\) to prove that~\(\rho_{t}\) has infinite order (see Fig.~\ref{fig:orbit}): put~\(\alpha_j =\size{\Orb_t(1^j)}\).
Since the activity is bounded, the set~\(\set{\sect{t^{\size{\Orb_t(1^j)}}}{1^j}}_j\) is finite, so there is a self-loop~\(\sect{t^{\alpha_i}}{1^i} = \sect{t^{\alpha_{i+j}}}{1^{i+j}}\). By putting  \(\sect{t^{\alpha_i}}{1^i} = ts_\alpha\) we get that~\(\sect{(ts_i)^\beta}{1^j} = ts_i\) for some integer~\(\beta=\size{\Orb_{ts_i}(1^{j})}\).
Suppose that the size of the orbit is 1. We obtain:
\vspace{-0.6cm}
\begin{center}
	\[
\begin{array}{ccc}
	&1^j	&			 \\
	
t	&\lacroix& t\\

	& 2^j	&		 \\	
s_i	&\lacroix& s_i\\

	&1^j	&		 \\		
 
\end{array}
\text{ but }
\begin{array}{ccc}
	&1	&			 \\
	
t	&\lacroix& t\\

	&2	&		 \\	
s_i	&\lacroix& \id\\

	&1	&		 \\		
 
\end{array}
\text{ so ~\(s_i = \id\) and }
\begin{array}{ccc}
	&1	&			 \\
	
t	&\lacroix& t\\

	&2	&		 \\	
\id	&\lacroix& \id\\

	&1	&		 \\		
 
\end{array}
\]
\end{center}
Hence~\(2=1\), contradiction. The size of the orbit of~\(\rho_t\)  increases strictly through the cycle~\(1^\beta\), whence~the order of \(\rho_t\)--the lowest common multiple of the  paths in the orbit signalizer of~\(t\) (\cite{BBSZ13})--is infinite.
\qed
\end{proof}

From this result, we get:
\begin{theorem}
	The probability that the group generated by an automaton with polynomial activity has an element of infinite order goes to~\(1\) as the size of the alphabet goes to infinity.
\end{theorem}
\begin{proof}
	We first prove that groups generated by automata with polynomial non finitary activity  generically  have an element of infinite order: let~\(\auta=\auttuple\) be an automaton in~\(\Pol(d)\setminus\Pol(-1)\) and let~\(t\) be a state with bounded activity on a nontrivial cycle. Since the activity does not depend on the choice of the production functions (except for the trivial state), we can consider the set~\(\mathcal{C}_\auta\) of automata in~\(\Pol(d)\setminus\Pol(-1)\) with same transition functions and trivial state. Among~\(\mathcal{C}_\auta\), we have~\(\rho_t(i)\neq i\) with probability~\(1-1/k\), so, in the normal form,~\(t\) is on a cycle labelled by~\(i\mot{x} \in \XX^\ell\) with~\(\rho_t(i\mot{x})\neq i\mot{x}\). We can apply Proposition~\ref{prop:boundorder}.\\
Now we show that the set~\(\Pol(-1)\)  has measure 0 in the set~\(\Pol(d), \: d \geq 0\) . If an automaton~\(\auta\) has polynomial activity, then there is at least one state~\(t\) satisfying~\(\delta_i(t) = \id\) for all~\(i\). Given~\(\auta \in \Pol(-1) \), we can build~\(k\) automata~\(\auta_i\) with bounded but not finitary activity by changing for exactly one letter~\(\delta_i(t) = \id\) to~\(\delta_i(t) = t\). If we consistently chose~\(t\) to be, \eg, the minimal among acceptable states, we can uniquely reconstruct~\(\auta\) from these~\(\auta_i\), whence the result.\\
We conclude using the law of total probability: the probability that an automaton in~\(\Pol(d)\) has an element of infinite order is equal to the probability that it has an element of infinite order given it belongs to~\(\Pol(d)\setminus \Pol(-1)\) times the probability of the later ; we showed that both goes to one, the result follows.
\qed
\end{proof}
From the proof we extract the following:
\begin{proposition}
		The probability that the group generated by a automaton in~\(\Pol(0)\)  on an alphabet of size~\(k\) has an element of infinite order is at least~\({k-1}\over{k+1}\).
\end{proposition}

\section{Conclusion and future work}
In this work, we proved, for various important classes of Mealy automata, that the generated groups have generically an element of infinite order, thus are infinite. It is natural to wonder whether other properties, such as non-amenability,  are generic and to extend these results to the full class of automaton group.

One interesting direction is to determine if generating a free or an infinitely presented group is generic in a class. These properties are mutually exclusive. Automata with polynomial activity cannot generate free groups~\cite{Nek10}, while reversible ones can~\cite{VoVo07} ; infinitely presented groups can be found in both classes.\\
It would be striking to find two classes and a group property which is nontrivial in both classes yet generically true in one and generically false in the other.\smallskip

\textit{Acknowledgements:} The author thanks Ville Salo who initiated this work and for interesting discussions, and Matthieu~Picantin and Jérémie~Brieussel for their comments. He was supported by Academy of Finland grant 296018 and  by the French \emph{Agence Nationale de la~Recherche}
through the project~AGIRA.


\bibliographystyle{amsplain}
\bibliography{biblio_main}

\end{document}